\documentclass[11pt]{amsart}
\usepackage{mathtools}
\usepackage{amsmath,amssymb,amsfonts,dsfont}
\usepackage{prettyref} 
\usepackage[utf8]{inputenc}
\usepackage{enumitem}
\usepackage[hyphens]{url}
\usepackage{tikz-cd}
\usepackage{tikz}
\usetikzlibrary{positioning}
\usepackage{MnSymbol}
\usepackage{bussproofs}
\usepackage[mathscr]{euscript}
\usepackage{hyperref}
\usepackage{amsthm}
\usepackage{theapa}
\usepackage{a4wide}
\usepackage[color=black,textcolor=white]{todonotes}
\usepackage{stackengine}
\usepackage{stmaryrd}
\usepackage{graphicx}
\usepackage{wrapfig}
\usepackage{float}
\graphicspath{{fig/}}
\usetikzlibrary{calc}
\usetikzlibrary{arrows,automata,topaths,matrix,positioning,fit}
\usepgflibrary{shapes.geometric}
\usetikzlibrary{shapes.geometric}
\tikzset{every state/.style={minimum size=0pt}}

\newtheorem{theorem}{Theorem}

\theoremstyle{definition} 
\newtheorem{definition}[theorem]{Definition}

\newtheorem{example}[theorem]{Example}

\newrefformat{cha}{Chapter \ref{#1}}
\newrefformat{sec}{Section \ref{#1}}
\newrefformat{tab}{Table \ref{#1}}
\newrefformat{fig}{Figure \ref{#1}}
\newrefformat{equ}{(\ref{#1})}
\newrefformat{app}{Appendix \ref{#1}}
\newrefformat{thm}{Theorem \ref{#1}}
\newrefformat{cor}{Corollary \ref{#1}}
\newrefformat{prop}{Proposition \ref{#1}}
\newrefformat{lem}{Lemma \ref{#1}}
\newrefformat{fact}{Fact \ref{#1}}
\newrefformat{obs}{Observation \ref{#1}}
\newrefformat{note}{Note \ref{#1}}
\newrefformat{idea}{Idea \ref{#1}}
\newrefformat{trivia}{Trivia \ref{#1}}
\newrefformat{def}{Definition \ref{#1}}
\newrefformat{not}{Notation \ref{#1}}
\newrefformat{con}{Convention \ref{#1}}
\newrefformat{rem}{Remark \ref{#1}}
\newrefformat{exa}{Example \ref{#1}}
\newrefformat{problem}{Problem \ref{#1}}
\newrefformat{claim}{Claim \ref{#1}}
\newrefformat{conjecture}{Conjecture \ref{#1}}
\newrefformat{exe}{Exercise \ref{#1}}
\newrefformat{alg}{Algorithm \ref{#1}}
\newrefformat{err}{Error \ref{#1}}
\newrefformat{que}{Question \ref{#1}}
\newrefformat{ite}{Item \ref{#1}}
\newrefformat{Q}{Q. \ref{#1}}
\newrefformat{warning}{Warning \ref{#1}}
\newrefformat{pseudocode}{Pseudocode \ref{#1}}

\newcommand{\I}{\mathcal I}
\newcommand{\yref}{\prettyref}
\newcommand{\ra}{\rightarrow}
\newcommand{\lfp}{\mathrm{lfp}\;}
\newcommand{\seq}{\subseteq}
\newcommand{\0}{\emptyset}
\newcommand{\tand}{\text{ and }}
\newcommand{\sm}{-}
\newcommand{\Ak}{\mathfrak A}
\newcommand{\s}{/_\sim}
\renewcommand{\dot}{\,.\,}
\newcommand{\Ir}{\I_R}
\newcommand{\vIp}{\I_{P_1}\dimes\I_{P_k}}
\newcommand{\Ipb}{\I_\Pb}
\newcommand{\PP}{\Psi_\Pb}
\newcommand{\dimes}{\times\ldots\times}
\newcommand{\ldimes}{\ltimes\ldots\ltimes}

\renewcommand{\d}{\delta}
\newcommand{\la}{\leftarrow}
\newcommand{\naf}{\,\mathrm{not}\,}
\newcommand{\abs}[1]{|\,#1\,|}
\newcommand{\Pk}{\mathfrak P}
\newcommand{\G}{\Gamma}
\newcommand{\Ip}{\mathcal{I}_P}
\newcommand{\Pp}{\Psi_P}
\newcommand{\Sk}{\mathfrak S}
\newcommand{\Rk}{\mathfrak R}
\newcommand{\A}{\mathcal A}
\newcommand{\Sb}{\mathbf{S}}
\newcommand{\Hb}{\mathbf{H}}
\newcommand{\Ib}{\mathbf{I}}
\newcommand{\Xb}{\mathbf X}
\newcommand{\Yb}{\mathbf Y}
\newcommand{\Tk}{\mathfrak T}
\newcommand{\Psn}{\Psi_{S_n}}
\newcommand{\pp}{\psi_\Pb}
\newcommand{\Pb}{\mathbf{P}}
\newcommand{\Ik}{(I_1,\ldots,I_k)}
\newcommand{\Jb}{\mathbf{J}}
\newcommand{\Ab}{\mathbf A}
\newcommand{\PA}{\Psi_\Ab}
\newcommand{\PB}{\Psi_\Bb}
\newcommand{\Bb}{\mathbf B}
\newcommand{\Tb}{\mathbf T}
\newcommand{\Cb}{\mathbf C}
\newcommand{\Prod}{P_1\ldimes P_k\,\left[\I_\Pb,\pp\right]}
\newcommand{\Iff}{\Leftrightarrow}
\renewcommand{\S}{\Sigma}
\renewcommand{\Pr}{\Psi_R}
\renewcommand{\P}{\mathcal P}

\title{On cascade products of answer set programs}
\author{
	Christian Anti\'c\\
}

\begin{document}
\begin{abstract} 
	Describing complex objects by elementary ones is a common strategy in mathematics and science in general. In their seminal 1965 paper, Kenneth Krohn and John Rhodes showed that every finite deterministic automaton can be represented (or ``emulated'') by a cascade product of very simple automata. This led to an elegant algebraic theory of automata based on finite semigroups (Krohn-Rhodes Theory). Surprisingly, by relating logic programs and automata, we can show in this paper that the Krohn-Rhodes Theory is applicable in Answer Set Programming (ASP). More precisely, we recast the concept of a cascade product to ASP, and prove that every program can be represented by a product of very simple programs, the reset and standard programs. Roughly, this implies that the reset and standard programs are the basic building blocks of ASP with respect to the cascade product. In a broader sense, this paper is a first step towards an algebraic theory of products and networks of nonmonotonic reasoning systems based on Krohn-Rhodes Theory, aiming at important open issues in ASP and AI in general.
\end{abstract}
\maketitle

\section{Introduction}

Describing complex objects by elementary ones is a common strategy in mathematics and science in general. For instance, the fundamental theorem of number theory states that every natural number can be (uniquely) represented by its prime factors. Similarly, in their seminal 1965 paper ``Algebraic theory of machines, I. Prime decomposition theorem for finite semigroups and machines'', Kenneth Krohn and John Rhodes showed that every finite deterministic automaton can be represented (or ``emulated'') by a cascade product of very simple automata. This led to an elegant algebraic theory of automata based on finite semigroups (Krohn-Rhodes Theory) and, more recently, to an algebraic theory of networks of automata \cite<cf.>{Domosi05}.

Answer Set Programming (ASP) \cite{Gelfond91}, on the other hand, has become a prominent knowledge representation and reasoning (KR\&R) formalism over the last two decades, with a wide range of applications in AI-related subfields such as, e.g., nonmonotonic reasoning, diagnosis, and planning \cite<cf.>{Brewka11}.  

In this paper, we aim at combining these two vivid areas of research and will show that, surprisingly, the Krohn-Rhodes Theory {\em is} applicable in ASP. More precisely, we recast the concept of a cascade product to ASP, and prove that every program can be represented by a product of reset programs $R=\{1\la\naf 1\}$ and $n$-standard programs $S_n$ consisting only of rules of the simple form $i\la j,\naf k$ (cf. \yref{thm:main}). Roughly, this implies that the reset and standard programs are the basic building blocks of ASP with respect to the cascade product and, strikingly, while the reset and standard programs do not possess any interesting declarative meaning (the reset program is inconsistent and the standard programs have only the empty answer set), their interaction can ``emulate'' any given program. In other words, the product semantics {\em emerges} from the interplay of its (simple) factors and allows for arbitrary complex behavior.

To the best of our knowledge, this is the first paper applying the Krohn-Rhodes Theory to logic programming. In a broader sense, it is a first step towards an algebraic theory of products and networks of nonmonotonic reasoning systems based on Krohn-Rhodes Theory, with far-reaching potential application areas including some important open issues in ASP and AI in general (cf. the discussion in \yref{sec:Discussion_and_Conclusion}).

The rest of the paper is structured as follows. In \yref{sec:preliminaries}, we present the basic definitions and results concerning ASP and automata. In \yref{sec:programmable}, we introduce the concept of a programmable automaton, and show that the distinguished reset and standard automata are programmable in this sense. In \yref{sec:homomorphic}, the main part of this paper, we recast the concept of a cascade product to ASP and prove that every program can be (homomorphically) represented by reset and standard programs. In \yref{sec:isomorphic}, we study the more restricted type of isomorphic representation and provide a complete class of programs with respect to it; moreover, we show that positive tight programs are isomorphically representable by reset programs. Finally, in \yref{sec:Discussion_and_Conclusion}, we conclude with a discussion on interesting lines for future research.

\section{Preliminaries}\label{sec:preliminaries}

We assume that the reader is familiar with the concept of a partially ordered set and that of a (complete) lattice. Following \cite{Gecseg86}, we denote by $[n]$, $n\geq 0$, the set $\{1,\ldots,n\}$. We denote, for $k\geq 1$ and $i\geq 0$, the least residue of $i$ modulo $n$ by $i\mod n$. For a set $X$, we denote by $\abs X$ the cardinality of $X$. Given a function $f:X\times Y\ra Z$, we denote by $f(\dot,y)$ the function from $X$ into $Z$ mapping each $x\in X$ to $f(x,y)\in Z$, and we denote by $\lfp f(\dot,y)$ the least fixpoint of $f(\dot,y)$. We denote the power set of $X$ by $\Pk(X)$. 

\subsection{Answer Set Programs}\label{sec:lp}

We briefly recall the syntax and answer set semantics \cite{Gelfond91} of nonmonotonic logic programs in an operator-based setting \cite<cf.>{Denecker00}.

\paragraph{Syntax} In the sequel, $\G$ will denote a finite nonempty set of propositional atoms. A ({\em normal logic}) {\em program} $P$ over some $\G_P$ is a finite nonempty set of rules of the form
\begin{align}\label{equ:rule} a\la b_1,\ldots,b_k,\naf b_{k+1},\ldots,\naf b_m,\quad m\geq k\geq 0,
\end{align} where $a,b_1,\ldots,b_m\in\G_P$ and $\naf$ denotes {\em negation-as-failure}. For convenience, we define for a rule $r$ of the form (\yref{equ:rule}), $H(r)=a$, $B^+(r)=\{b_1,\ldots,b_k\}$, $B^-(r)=\{b_{k+1},\ldots,b_m\}$, and $B(r)=B^+(r)\cup B^-(r)$. We call $r$ a {\em fact}, if $B(r)=\0$; and we call $r$ {\em positive} if $B^-(r)=\0$. We say that $P$ is {\em positive} if every rule $r\in P$ is positive, and we call $P$ {\em tight} if there is a mapping $\ell$ from $\G_P$ into the nonnegative integers such that for each rule $r$ in $P$, $\ell(H(r))>\ell(b)$ for every $b\in B^+(r)$.

\paragraph{Semantics} An {\em interpretation} of $P$ is any subset $I\seq \G_P$ and we denote the set of all interpretations of $P$ by $\Ip=\Pk(\G_P)$. Define the {\em 4-valued immediate consequence operator} $\Pp:\Ip\times\Ip\ra\Ip$ by $$\Pp(I,J)=\{H(r) : r\in P, B^+(r)\seq I, B^-(r)\cap J=\0\}.$$ Intuitively, $\Pp(I,J)$ contains the heads $H(r)$ of all rules $r$ in $P$ where the positive part of the body evaluates to {\em true} in $I$, and the negative part evaluates to {\em true} in $J$. Given some $I\in\Ip$, it is well-known that $\Pp(\dot,I)$ is monotone on the complete lattice $\Ip$ ordered by $\seq$, and hence has a {\em least fixpoint} denoted by $\lfp\Pp(\dot,I)$. We say that $I\in\Ip$ is an {\em answer set} of $P$, or a {\em $\Pp$-answer set}, if $I=\lfp\Pp(\dot,I)$.

\subsection{Krohn-Rhodes Theory}\label{sec:automata}

\begin{figure}[t]
\begin{center} 
\begin{tikzpicture}[->,>=stealth',semithick,shorten >= 1pt,node distance=2cm,auto]
\node[state] (1) {$1$};
\node[state] (2) [right of=1] {$2$};

\path[->]   
    (1) edge [loop above] node {$\sigma_0$} (1)
    (1) edge [bend left] node {$\sigma_1$} (2)
    (2) edge [loop above] node {$\sigma_1$} (2)
    (2) edge [bend left] node {$\sigma_0$} (1)
    ;
\end{tikzpicture}
\hspace{1cm}
\begin{tikzpicture}[->,>=stealth',semithick,shorten >= 1pt,node distance=1.5cm,auto]
\node[state] (1) {$1$};
\node[state] (2) [above right of=1] {$2$};
\node[state] (3) [right of=2] {$3$};
\node[state] (4) [below right of=3] {};
\node[state] (5) [below left of=4] {};
\node[state] (n) [below right of=1] {$n$};

\path[->]
  (1) edge [loop left] node {$\sigma_0$} (1)
  (2) edge [loop above] node {$\sigma_0$} (2)
  (3) edge [loop above] node {$\sigma_0,\sigma_2$} (3)
  (4) edge [loop right] node {$\sigma_0,\sigma_2$} (4)
  (5) edge [loop below] node {$\sigma_0,\sigma_2$} (5)
  (n) edge [loop below] node {$\sigma_0,\sigma_2$} (n)
  (1) edge [bend left] node {$\sigma_1,\sigma_2$} (2)
  (2) edge [bend left] node {$\sigma_2$} (1)
  (2) edge node {$\sigma_1$} (3)
  (3) edge [dashed] node {$\sigma_1$} (4)
  (4) edge [dashed] node {$\sigma_1$} (5)
  (5) edge [dashed] node {$\sigma_1$} (n)
  (n) edge node {$\sigma_1$} (1)
  ;
\end{tikzpicture}
\end{center}
\caption{The (two-state) reset automaton $\Rk$ and the $n$-standard automaton $\Sk_n$.}
\label{fig:automata}
\end{figure}

In this section, we recall some basic definitions and results of Krohn-Rhodes Theory by mainly following the lines of \cite[Chapters 1--3]{Gecseg86}.

An {\em automaton} $\Ak=(Q,\S,\delta)$ consists of a finite set $Q$ of {\em states}, a finite nonempty set $\S$, called the {\em input alphabet}, and a mapping $\d:Q\times\S\ra Q$ called the {\em transition function}.

Given two automata $\Ak=(Q,\S,\delta)$ and $\Ak'=(Q',\S',\delta')$, we say that $\Ak'$ is a {\em subautomaton} of $\Ak$ if $Q'\seq Q$, $\S'\seq\S$, and $\delta'$ is the restriction of $\delta$ to $Q'\times\S'$. A pair $h=(h_1,h_2)$ of surjective mappings $h_1:Q\ra Q'$, $h_2:\S\ra\S'$ is a {\em homomorphism} of $\Ak$ onto $\Ak'$ if $h_1(\delta(q,x))=\delta'(h_1(q),h_2(x))$, for every $q\in Q, x\in\S$. The pair $h$ is an {\em isomorphism} if $h_1$ and $h_2$ are bijective homomorphisms, and we say that $\Ak$ is {\em isomorphic} to $\Ak'$ if there exists an isomorphism $h$ of $\Ak$ onto $\Ak'$. If $\S=\S'$, then we omit $h_2$ and define $h=h_1$.

An equivalence relation $\sim$ on $Q$ is a {\em congruence relation} of $\Ak$ if $q\sim q'$ implies $\delta(q,x)\sim\delta(q',x)$, for all $q,q'\in Q$ and $x\in\S$. We denote the congruence class of $q\in Q$ with respect to $\sim$ by $q\s$, and define the {\em quotient automaton} $\Ak\s=(Q\s,\S,\delta\s)$ by $\delta\s(q\s,x)=\delta(q,x)\s$ for all $q\in Q$ and $x\in \S$. Conversely, given a homomorphism $h=(h_1,h_2)$ of $\Ak$ onto $\Ak'$, we mean by the {\em congruence relation of $\Ak$ induced by $h$} the binary relation $\sim$ on $Q$ given by $q\sim q'$ if $h_1(q)=h_1(q').$

The following automata will play a central role throughout the rest of the paper (cf. \yref{fig:automata}):
\begin{enumerate} 
\item Define the ({\em two-state}) {\em reset automaton} $\Rk=([2],\{\sigma_0,\sigma_1\},\delta_\Rk)$ by $\delta_\Rk(i,\sigma_0)=1$, and $\delta_\Rk(i,\sigma_1)=2$, for all $i\in [2]$.
\item We call an automaton $\Sk=([n],\{\sigma_0,\sigma_1,\sigma_2\},\delta_\Sk)$, $n>1$, {\em standard} if $\delta_\Sk$ satisfies the following conditions, for all $i\in [n]$:
\begin{enumerate} 
\item $\delta_\Sk(i,\sigma_0)=i$;
\item $\delta_\Sk(i,\sigma_1)= (i\mod n)+1$;
\item $\delta_\Sk(i,\sigma_2)= \begin{cases} 
                      2 & \text{if }i=1, \\
                      1 & \text{if }i=2, \\
                      i & \text{otherwise.}
                      \end{cases}$
\end{enumerate} We denote the $n$-state standard automaton by $\Sk_n$.
\end{enumerate}

The following operators on arbitrary classes $\A$ of automata will be useful:
\begin{enumerate}
\item $\Sb(\A)$ denotes the set of subautomata of automata from $\A$;
\item $\Hb(\A)$ denotes the homomorphic images of automata from $\A$;
\item $\Ib(\A)$ denotes the isomorphic images of automata from $\A$.
\end{enumerate} We will write $\Xb\Yb(\A)$ for $\Xb(\Yb(\A))$, where $\Xb$ and $\Yb$ are operators from above.

We now define the cascade product for automata, which is also known as the wreath \cite{Krohn65} or $\alpha_0$-product \cite{Gecseg86} in the literature.

\begin{definition}[Cascade Automata Product]\label{def:cascade_automata} For some $k>0$, let $\Ak_i=(Q_i,\S_i,\delta_i)$, $i\in [k]$, be a family of automata, and let $\S$ be an alphabet. A {\em feedforward function} for $\Ak_1,\ldots,\Ak_k$ is a mapping $\psi:(Q_1\dimes Q_k)\times\S\ra\S_1\dimes\S_k$ with $$\psi((q_1,\ldots,q_k),\sigma)=(\psi_1((q_1,\ldots,q_k),\sigma),\ldots,\psi_k((q_1,\ldots,q_k),\sigma))$$ where the {\em component feedforward function} $\psi_i$, $i\in [k]$, is a mapping from $(Q_1\dimes Q_k)\times\S$ into $\S_i$. In the sequel, we omit those arguments $q_j$, $j\in [k]$, $\psi_i$ does not depend on. The {\em cascade} (or {\em loop-free}) {\em automata product} of $\Ak_1,\ldots,\Ak_k$ with respect to $\S_\Ak=\S$ and some feedforward function $\psi_\Ak$ $$\Ak=(Q_\Ak,\S_\Ak,\delta_\Ak)=\Ak_1\ldimes\Ak_k\,[\S_\Ak,\psi_\Ak]$$ is given by $Q_\Ak=Q_1\dimes Q_k$ where $\psi_i$, $i\in [k]$, is independent of its $j^{\text{th}}$ component, $j\in [k]$, whenever $j\geq i$. Finally, we define the {\em transition function} $\delta_\Ak:Q_\Ak\times\S_\Ak\ra Q_\Ak$ by $$\delta_\Ak((q_1,\ldots,q_k),\sigma)=(\delta_1(q_1,\psi_1(\sigma)),\ldots,\delta_k(q_k,\psi_k((q_1,\ldots,q_{k-1}),\sigma))).$$
\end{definition}

\begin{definition}\label{def:representation-automata} We say that an automaton $\Ak$ {\em homomorphically} (resp., {\em isomorphically}) {\em represents} an automaton $\Ak'$ if $\Ak'\in\Hb\Sb(\{\Ak\})$ (resp., $\Ak'\in\Ib\Sb(\{\Ak\})$). Moreover, we say that a class $\A$ of automata is {\em homomorphically} (resp., {\em isomorphically}) {\em complete} with respect to the cascade automata product if every automaton $\Ak$ can be homomorphically (resp., isomorphically) represented by a cascade automata product of automata from $\A$.
\end{definition}

The following result is a consequence of the Krohn-Rhodes decomposition theorem \cite{Krohn65}, and it will be of great importance for our main \yref{thm:main}.

\begin{theorem}\label{thm:Krohn-Rhodes}[cf. \cite{Gecseg86}, Theorem 2.1.5] Let $\Ak$ be an automaton with $n>1$ states. Then, $\Ak$ can be homomorphically represented by a cascade automata product of reset and $n$-state standard automata over the same input alphabet as $\Ak$.
\end{theorem}

We now turn to isomorphic completeness. Let $\Tk_n=([n],\S_n,\delta_n)$, $n\geq 1$, such that $\S_n$ is the set of all mappings $\sigma:[n]\ra [n]$, and $\delta_n(j,\sigma)=\sigma(j)$, for all $j\in [n]$.

\begin{theorem}\label{thm:Tkn}[cf. \cite{Gecseg86}, Theorem 3.2.1] A class $\A$ of automata is isomorphically complete with respect to the cascade automata product iff for every $n\geq 1$, there exists some $\Ak\in\A$ such that $\Tk_n$ can be embedded into a cascade automata product $\Ak\,[\S,\psi]$, consisting of a single factor.
\end{theorem}

\section{Programmable Automata}\label{sec:programmable}

\begin{figure}[t]
\begin{center} 
\begin{tikzpicture}[->,>=stealth,semithick,shorten >= 1pt,node distance=3cm,auto]
\node[state] (q_0) {$\0$};
\node[state] (q_1) [above left of=q_0] {$\{a\}$};
\node[state] (q_2) [above right of=q_0] {$\{b\}$};
\node[state] (q_3) [above right of=q_1] {$\{a,b\}$};

\path[->]   
      (q_0) edge [loop below] node {$\{a,b\}$} (q_0)
      (q_0) edge node [below=.1] {$\{a\}$} (q_1)
      (q_0) edge node [below=.1] {$\{b\}$} (q_2)
      (q_0) edge [bend left] node {$\0$} (q_3)
      (q_1) edge [loop left] node {$\{a\}$} (q_1)
      (q_1) edge node [above=.1] {$\0$} (q_3)
      (q_1) edge [bend left] node {$\{b\}$} (q_2)
      (q_1) edge [bend right] node [below=.3] {$\{a,b\}$} (q_0)
      (q_2) edge [loop right] node {$\{b\}$} (q_2)
      (q_2) edge [bend left] node [below=.3] {$\{a,b\}$} (q_0)
      (q_2) edge node [above=.1] {$\0$} (q_3)
      (q_2) edge [bend left] node {$\{a\}$} (q_1)
      (q_3) edge [loop above] node {$\0$} (q_3)
      (q_3) edge [bend left] node {$\{a,b\}$} (q_0)
      (q_3) edge [bend left] node {$\{b\}$} (q_2)
      (q_3) edge [bend right] node [above=.1] {$\{a\}$} (q_1)
      ;
\end{tikzpicture}
\end{center}
\caption{The characteristic automaton $\Psi_B$ of the program $B=\{a\la\naf b; b\la\naf a\}$.}
\label{fig:B}
\end{figure}

In this section, we relate programs and automata and prove in \yref{thm:programmable} that the distinguished automata given in \yref{sec:automata} can be ``realized'' by programs. This connection will serve as the basis for the rest of the paper, and for the main \yref{thm:main} in particular.

Given some program $P$, we define its {\em characteristic automaton} $\Ak_P=(Q_P,\S_P,\delta_P)$ by $Q_P=\S_P=\Ip$ and $\delta_P=\Pp$. In the sequel, we will not distinguish between the operator $\Pp$ and the characteristic automaton $\Ak_P=(\Ip,\Ip,\Pp)$, i.e., we will refer to $\Ak_P$ simply by $\Pp$ and will call $\Pp$ the characteristic automaton of $P$ (cf. \yref{fig:B}).

\begin{definition}\label{def:programmable} We say that an automaton $\Ak$ is {\em homomorphically} (resp., {\em isomorphically}) {\em programmable} if there exists some program $P$ such that $\Pp$ homomorphically (resp., isomorphically) represents $\Ak$, that is, $\Ak\in\Hb\Sb(\{\Pp\})$ (resp., $\Ak\in\Ib\Sb(\{\Pp\})$). We then say that $P$ {\em homomorphically} (resp., {\em isomorphically}) {\em programs} $\Ak$.
\end{definition}

We illustrate this concept with an example; in \yref{thm:programmable} we will see that the reset automaton $\Rk$ and the $n$-state standard automaton $\Sk_n$, $n>1$, are isomorphically programmable.

\begin{example}\label{exa:elevator} Define the {\em elevator automaton} $\mathfrak E=([2],\{\sigma_0,\sigma_1\},\delta_\mathfrak E)$ by $\delta_\mathfrak E(1,\sigma_0)=1$, $\delta_\mathfrak E(1,\sigma_1)=2$, and $\delta_\mathfrak E(2,\sigma_0)=\delta_\mathfrak E(2,\sigma_1)=2$ \cite<cf.>[p.45]{Domosi05}. On the other hand, define the {\em elevator program} $E$ by $E=\{e\la e; e\la\naf e\}$. Then, $h=(h_1,h_2)$ defined by $h_1(\0)=1$, $h_1(\{e\})=2$, $h_2(\{e\})=\sigma_0$, and $h_2(\0)=\sigma_1$ is an isomorphism of (the automaton) $\Psi_E$ onto $\mathfrak E$; hence, $E$ isomorphically programs $\mathfrak E$.
\end{example}

For convenience, in the sequel we occasionally denote atoms by nonnegative integers.

\begin{definition}\label{def:programs} The {\em reset program} $R$ over $\G_R=[1]$ consists of the following single rule: $$1\la\naf 1.$$ The {\em $n$-standard program} (or {\em $n$-program}) $S_n$ over $\G_n=[n]\cup\{3\}$, $n>1$, consists of the following rules, for all $i\in [n]$ and $j\in [n]$, $j>2$:
\begin{align*} 
i &\la i,\;\naf 1,         & 1 &\la 2,\;\naf 3,  \\
(i\mod n)+1 &\la i,\;\naf 2, & 2 &\la 1,\;\naf 3,  \\
&                         & j &\la j,\;\naf 3.
\end{align*}
\end{definition}

Note that the reset program $R$ is {\em inconsistent}, i.e., has no $\Pr$-answer sets, and for every $n>1$, the $n$-program $S_n$ has the $\Psn$-answer set $\0$.

\begin{theorem}\label{thm:programmable} The reset program $R$ and the $n$-standard program $S_n$ isomorphically program the reset automaton $\Rk$ and the $n$-state standard automaton $\Sk_n$, $n>1$, respectively.
\end{theorem}
\begin{proof} Define $h_{R,1}:\Ir\ra[2]$ and $h_{R,2}:\Ir\ra\{\sigma_0,\sigma_1\}$ by $h_{R,1}(\0)=1$, $h_{R,1}(\{1\})=2$, $h_{R,2}(\0)=\sigma_1$, and $h_{R,2}(\{1\})=\sigma_0$. A straightforward computation shows that $h_R=(h_{R,1},h_{R,2})$ is an isomorphism of $\Pr$ onto $\Rk$; i.e., we have $$h_{R,1}(\Pr(I,J))=\delta_\Rk(h_{R,1}(I),h_{R,2}(J)),\quad\text{for all }I,J\in\Ir.$$ Hence, $\Rk\in\Ib\Sb(\{\Pr\})$.

For the second part, let $\Psn'=(\I_{S_n}',\I_{S_n}'',\Psn')$ be the subautomaton of $\Psn$ given by $\I_{S_n}'=\{\{i\}:i\in [n]\}\seq\I_{S_n}$, $\I_{S_n}''=\{\{2,3\},\{1,3\},\{1,2\}\}\seq\I_{S_n}$, and $\Psn'$ equals $\Psn$ restricted to $\I_{S_n}'\times\I_{S_n}''$. Define $h_{S_n,1}:\I_{S_n}'\ra [n]$ by $h_{S_n,1}(\{i\})=i$, for all $i\in [n]$; and $h_{S_n,2}:\I_{S_n}''\ra\{\sigma_0,\sigma_1,\sigma_2\}$ by $h_{S_n,2}(\{2,3\})=\sigma_0$, $h_{S_n,2}(\{1,3\})=\sigma_1$, and $h_{S_n,2}(\{1,2\})=\sigma_2$. Then, $h=(h_{S_n,1},h_{S_n,2})$ is an isomorphism of $\Psn'$ onto $\Sk_n$; i.e., we have $$h_{S_n,1}(\Psn'(\{i\},J))=\delta_{\Sk_n}(h_{S_n,1}(\{i\}),h_{S_n,2}(J)),\quad\text{for all }i\in [n]\tand J\in\I_{S_n}''.$$ Hence, $\Sk_n\in\Ib\Sb(\{\Psn\})$.
\end{proof}

\section{Cascade Products and Homomorphic Representations}\label{sec:homomorphic}


In this section, we recast the concept of a cascade automata product presented in \yref{sec:automata} (cf. \yref{def:cascade_automata}) to the setting of ASP and study homomorphic representations.

\begin{definition}[Cascade Program Product]\label{def:cascade} Let $P_1,\ldots,P_k$, $k>1$, be a family of programs over some alphabets $\G_{P_1},\ldots,\G_{P_k}$, respectively, and let $\I_\Pb$ be some finite nonempty set. A {\em feedforward function} for $P_1,\ldots,P_k$ is a mapping $\pp:(\vIp)\times\Ipb\ra\vIp$ with $$\pp(\Ik,\Jb)=(\psi_{\Pb,1}(\Ik,\Jb),\ldots,\psi_{\Pb,k}(\Ik,\Jb))$$ where the {\em component feedforward function} $\psi_{\Pb,i}$, $i\in [k]$, is a mapping from $(\vIp)\times\Ipb$ into $\I_{P_i}$. In the sequel, we omit those arguments $I_j$, $j\in [k]$, $\psi_{\Pb,i}$ does not depend on. The ({\em cascade} or {\em loop-free program}) {\em product} of $P_1,\ldots,P_k$ with respect to $\I_\Pb$ and some feedforward function $\pp$
\begin{align*} \Pb=\Prod
\end{align*} is given by its component feedforward functions $\psi_{\Pb,i}$, $i\in [k]$, which are independent of their $j^{\text{th}}$ component, $j\in [k]$, whenever $j\geq i$. Finally, we define the {\em characteristic automaton} $\PP=(Q_\Pb,\S_\Pb,\PP)$ of $\Pb$ by $Q_\Pb=\vIp$, $\S_\Pb=\Ipb$, and $\PP:(\vIp)\times\Ipb\ra\vIp$ with $$\PP(\Ik,\Jb)=(\Psi_{P_1}(I_1,\psi_{\Pb,1}(\Jb)),\ldots,\Psi_{P_k}(I_k,\psi_{\Pb,k}((I_1,\ldots,I_{k-1}),\Jb))).$$
\end{definition}
 
Intuitively, a cascade program product is a collection of programs which are connected to each other and exchange (local) information via a feedforward function, where each component program may depend only on the preceding components and on the global input; every state-transition of the characteristic automaton of the product is then the result of the {\em simultaneous local} state-transitions of the characteristic automata of its component programs.

Formally, a product is not a program according to the definition given in \yref{sec:lp}. However, we can relate products and programs as follows (cf. \yref{def:representation-automata}).

\begin{definition}\label{def:representation} We say that a cascade program product $\Pb$ {\em homomorphically} (resp., {\em isomorphically}) {\em represents} a program $P$ if $\PP$ homomorphically (resp., isomorphically) represents $\Pp$, that is, $\Pp\in\Hb\Sb(\{\PP\})$ (resp., $\Pp\in\Ib\Sb(\{\PP\})$). Moreover, we say that a class $\P$ of programs is {\em homomorphically} (resp., {\em isomorphically}) {\em complete} with respect to the cascade program product if every program $P$ can be homomorphically (resp., isomorphically) represented by a cascade program product of programs from $\P$.
\end{definition}

We now make the relation between products and programs more explicit. In the context of logic programming, representation (or ``emulation'') means semantic equivalence (modulo some encoding). According to \yref{def:representation}, a product $\Pb=P_1\ldimes P_k\,[\Ipb,\pp]$, $k>1$, represents a program $P$ if the characteristic automaton $\PP$ represents the characteristic automaton $\Pp$ (in the sense of \yref{sec:automata}); that is, if there exists a subautomaton $\PP'=(\I'_{P_1}\dimes\I'_{P_k},\Ipb',\PP')$ of $\PP$ and a congruence relation $\sim$ on $\I'_{P_1}\dimes\I'_{P_k}$ such that $\PP'\s$ is isomorphic to $\Pp$. Intuitively, every interpretation $I\in\Ip$ of $P$ then corresponds to a congruence class of $k$-tuples from $\I'_{P_1}\dimes\I'_{P_k}$; if the representation is isomorphic, then $I$ can be identified with a single $k$-tuple $(I'_1,\ldots,I'_k)$ and in this case we can imagine $(I'_1,\ldots,I'_k)$ to be an ``encoding'' of $I$.

Interestingly enough, by the forthcoming \yref{thm:main}, we can assume that only reset and standard programs occur as factors in the product $\Pb$. That is, \yref{thm:main} roughly implies that by knowing the reset program $R$ and all the $n$-programs $S_n$, $n>1$, and by knowing how to form the cascade program product, we essentially know all programs; viz., the reset and standard programs are the basic building blocks of ASP with respect to the cascade program product.

We are now ready to state the main theorem of this paper.

\begin{theorem}\label{thm:main} Every program $P$ over some alphabet $\G_P$, with $\abs{\G_P}=m$, can be homomorphically represented by a cascade program product $\Pb$ of reset and $2^m$-standard programs.
\end{theorem}
\begin{proof} According to \yref{def:representation}, we have to show that there exists some product $\Pb$ such that $\PP$ homomorphically represents $\Pp$. Since $\Pp$ has $2^m$ states, \yref{thm:Krohn-Rhodes} yields a cascade automata product $\Ak_P=\Ak_1\ldimes\Ak_k\,[\Ip,\psi_P]$, for some $k>0$, consisting of reset and $2^m$-standard automata homomorphically representing $\Pp$. Note that $\Ak_P$ has the same input alphabet $\Ip$ as $\Pp$. Define the product $\Pb=P_1\ldimes P_k\,[\Ipb,\pp]$ as follows: (i) for every $i\in [k]$, if $\Ak_i$ is the reset automaton $\Rk$ (resp., $2^m$-standard automaton $\Sk_{2^m}$), then $P_i$ is the reset program $R$ (resp., $2^m$-standard program $S_{2^m}$); (ii) $\Ipb$ is the input alphabet $\Ip$ of $\Ak_P$ and $\Pp$; (iii) $\pp$ is a mapping from $(\vIp)\times\Ip$ into $\vIp$ where $\I_{P_i}$, $i\in [k]$, is $\Ir$ (resp., $\I_{2^m}$) if $P_i$ is the reset program $R$ (resp., $2^m$-standard program $S_{2^m}$), and $\psi_{\Pb,i}$ coincides with $\psi_{P,i}$ on the appropriate subset of $\vIp$ modulo the isomorphisms defined in the proof of \yref{thm:programmable}. Then, it follows from \yref{thm:programmable} that $\PP$ isomorphically represents $\Ak_P$ and, by transitivity of representation, it homomorphically represents $\Pp$, which proves our theorem.
\end{proof}

It is worth noting that the proof of \yref{thm:main} yields a product $\Pb$ whose characteristic automaton $\PP$ has the same input alphabet $\Ip$ as the characteristic automaton $\Pp$ of $P$. Therefore, we can characterize the answer sets of $P$ by $\PP$ as follows. Roughly, the product semantics of $\Pb$ {\em emerges} as an interaction of its (simple) factors $P_1,\ldots,P_k$ with respect to $P$. More precisely, by the remarks given above, there exists a quotient subautomaton $\PP'\s$ of $\PP$ which is isomorphic to $\Pp$ and which has the same input alphabet as $\Pp$. Let $h:\I'_{P_1}\dimes\I'_{P_k}\ra\Ip$ be the corresponding homomorphism of $\PP'$ onto $\Pp$ inducing $\sim$; we order $(\I'_{P_1}\dimes\I'_{P_k})\s$ by $(I_1,\ldots,I_k)\s\seq_h (I_1',\ldots,I_k')\s$ if $h(I_1,\ldots,I_k)\seq h(I_1',\ldots,I_k')$. Then, $((\I'_{P_1}\dimes\I'_{P_k})\s,\seq_h)$ is isomorphic (as a lattice) to $(\Ip,\seq)$, and we say that $I\in\Ip$ is a {\em $\PP'\s$-answer set} if $I=h(\lfp\PP'\s(\dot,I))$. Then, we have the following correspondence:
\begin{align}\label{equ:answer-set}\text{$I$ is a $\Pp$-answer set}\Iff\text{$I$ is a $\PP'\s$-answer set.}
\end{align} By \yref{thm:main}, we can assume that in the right hand side of (\ref{equ:answer-set}), only reset and $2^m$-standard programs occur. 

We illustrate these concepts by giving some examples.

\begin{example}\label{exa:A} Let $A=\{a\la\}$ be a program consisting of a single fact. We can interpret $A$ as a database {\em storing} some information represented by $a$. Observe that neither the reset program $R$ nor the $2$-program $S_2$ contains a {\em fact}. However, we verify that $$\Ab=R\,[\I_A,\psi_\Ab]=\{1\la\naf 1\}\,[\I_A,\psi_\Ab]$$ defined by $\psi_\Ab(J)=\0$, for all $J\in\I_A$, isomorphically represents $A$. Define $h:\Ir\ra\I_A$ by $h(\0)=\0$ and $h(\{1\})=\{a\}$. We check that $h$ is an isomorphism:
\begin{align*} 
h(\PA(I,J)) &= h(\Pr(I,\psi_\Ab(J)) \\
            &= h(\Pr(I,\0)) \\
            &= h(\{1\}) \\
            &= \{a\} \\
            &= \Psi_A(h(I),J)
\end{align*} holds for all $I\in\Ir$ and $J\in\I_A$. Therefore, the congruence relation $\sim$ induced by $h$ is the trivial diagonal relation and $\PA\s$ is isomorphic to $\PA$. Hence, $\Psi_A\in\Ib\Sb(\{\PA\})$. The calculation above proves that $\{a\}$ is the only $\Psi_A$-answer set or, equivalently, the only $\PA$-answer set. Intuitively, $\Ab$ ``emulates'' the storage of the fact $a$ by ignoring the input $J$ appropriately. Generally, the program $A_m=\{a_1\la;\ldots;a_m\la\}$, $m\geq 1$, is isomorphically represented by $\Ab_m=R\ldimes R\,[\I_{A_m},\psi_{\Ab_m}]$ (with $m$ factors) where $\psi_{\Ab_m,i}((I_1,\ldots,I_{i-1}),J)=\0$, for all $i\in [m]$, $I_1,\ldots,I_{i-1}\in\Ir$, and $J\in\I_{A_m}$. Here, an isomorphism is an arbitrary ``binary encoding'' $h$ of $\I_{A_m}$; e.g., $h(I_1,\ldots,I_m)=\{a_i\in\I_{A_m} : I_i=\{1\}, i\in [m]\}$. 
\end{example}

\begin{example}\label{exa:B} The program $B=\{a\la\naf b; b\la\naf a\}$ (cf. \yref{fig:B}) is isomorphically represented by the cascade program product $$\Bb=R\ltimes R\,[\I_B,\psi_\Bb]=\{1\la\naf 1\}\ltimes\{1\la\naf 1\}\,[\I_B,\psi_\Bb]$$ defined by
\begin{align*}
\psi_{\Bb,1}(\0) &= \psi_{\Bb,1}(\{a\})=\0, & \psi_{\Bb,2}(I,\0) &= \psi_{\Bb,2}(I,\{b\})=\0, \\
\psi_{\Bb,1}(\{b\}) &= \psi_{\Bb,1}(\{a,b\})=\{1\}, & \psi_{\Bb,2}(I,\{a\}) &= \psi_{\Bb,2}(I,\{a,b\})=\{1\},
\end{align*} for all $I\in\Ir$. Let $h:\Ir\times\Ir\ra\I_B$ be the ``binary encoding'' of $\I_B$ given by $h(\0,\0)=\0$, $h(\{1\},\0)=\{a\}$, $h(\0,\{1\})=\{b\}$, and $h(\{1\},\{1\})=\{a,b\}$. It is straightforward to verify that $h$ is an isomorphism of $\PB$ onto $\Psi_B$. For instance, we compute:
\begin{align*} 
h(\PB((\0,\0),\{a\})) &= h(\Pr(\0,\psi_{\Bb,1}(\{a\})),\Pr(\0,\psi_{\Bb,2}(\0,\{a\}))) \\
                      &= h(\Pr(\0,\0),\Pr(\0,\{1\})) \\
                      &= h(\{1\},\0) \\
                      &= \{a\} \\
                      &= \Psi_B(h(\0,\0),\{a\}).
\end{align*} Hence, $\Psi_B\in\Ib\Sb(\{\PB\})$. By the remarks given above, $I$ is a $\Psi_B$-answer set iff $I$ is a $\PB$-answer set and, clearly, $\{a\}$ and $\{b\}$ are the only ones.
\end{example}

\section{Isomorphic Representations}\label{sec:isomorphic}
In this section, we study the more restricted type of isomorphic representation and provide a complete class of programs with respect to it. Moreover, in \yref{thm:positive} we show that every positive tight program can be isomorphically represented by a cascade program product of reset programs.

For some $n\geq 1$, let $\sigma_1,\ldots,\sigma_{n^n}$ be an enumeration of the set of all mappings from $[n]$ into $[n]$. Define $T_n$ over $\G_{T_n}=[n^n]$ to be the program consisting of the rules, for all $j\in [n]$ and $k\in [n^n]$:
\begin{align*} \sigma_k(j)&\la j,\,\naf k.
\end{align*} As a consequence of \yref{thm:Tkn}, we obtain the following completeness result.

\begin{theorem}\label{thm:isomorphic} The class of programs consisting of all $T_n$, $n\geq 1$, is isomorphically complete with respect to the cascade program product.
\end{theorem}
\begin{proof} According to \yref{thm:Tkn} and \yref{def:representation}, we have to show that for every $n\geq 1$, the automaton $\Tk_n=([n],\S_n,\delta_n)$ can be embedded into a cascade automata product of $\Psi_{T_n}$ with a single factor. Define $\Psi_{\Tb_n}=\Psi_{T_n}\,[\I_{T_n},\psi_{\Tb_n}]$ by $\psi_{\Tb_n}(J)=J$, for all $J\in\I_{T_n}$. Define the embedding $h=(h_1,h_2)$, with $h_1:[n]\ra\I_{T_n}$ and $h_2:\S_n\ra\I_{T_n}$, by $h_1(j)=\{j\}$ and $h_2(\sigma_k)=\{1,\ldots,k-1,k+1,\ldots,n^n\}$, $k\in [n^n]$. Clearly, $h_1$ and $h_2$ are one-one, and the following computation proves that $h$ is indeed an embedding:
\begin{align*} 
\Psi_{\Tb_n}(h_1(j),h_2(\sigma_k)) &= \Psi_{T_n}(h_1(j),\psi_{\Tb_n}(h_2(\sigma_k))) \\
                        &= \Psi_{T_n}(h_1(j),h_2(\sigma_k)) \\
                        &= \Psi_{T_n}(\{j\},\{1,\ldots,k-1,k+1,\ldots,n^n\}) \\
                        &= \{\sigma_k(j)\} \\
                        &= h_1(\sigma_k(j)) \\
                        &= h_1(\delta_n(j,\sigma_k))
\end{align*} holds for all $j\in [n]$ and $k\in [n^n]$.
\end{proof}

We now turn to the restricted class of positive (i.e., negation-free) tight programs.

\begin{example}\label{exa:C} Consider the positive tight program $C=\{a\la; b\la a; c\la a,b\}$. The product $\Cb=R\ltimes R\ltimes R\,[\I_C,\psi_\Cb]$ given by 
\begin{align*} 
\psi_{\Cb,1}(J) = \0 && \psi_{\Cb,2}(I_1,J) = \{1\}\sm I_1 && \psi_{\Cb,3}((I_1,I_2),J) = \{1\}\sm (I_1\cap I_2)
\end{align*} for all $I_1,I_2\in\Ir$ and $J\in\I_C$, isomorphically represents $C$. Again, we define the isomorphism $h$ to be a ``binary encoding'' of $\I_C$ where, e.g., $(\{1\},\0,\0)$ is mapped to $\{a\}$, $(\{1\},\0,\{1\})$ is mapped to $\{a,c\}$ and so on. For instance, we can compute the least model $I=\{a,b,c\}$ of $C$ as follows:
\begin{align*} 
&h(\Psi_\Cb((\0,\0,\0),J)) = h(\Pr(\0,\0),\Pr(\0,\{1\}),\Pr(\0,\{1\})) = h(\{1\},\0,\0) = \{a\} \\
&h(\Psi_\Cb((\{1\},\0,\0),J)) = h(\Pr(\{1\},\0),\Pr(\0,\0),\Pr(\0,\{1\})) = h(\{1\},\{1\},\0) = \{a,b\} \\
&h(\Psi_\Cb((\{1\},\{1\},\0),J)) = h(\Pr(\{1\},\0),\Pr(\{1\},\0),\Pr(\0,\0)) = h(\{1\},\{1\},\{1\}) = I \\
&h(\Psi_\Cb((\{1\},\{1\},\{1\}),J)) = h(\Pr(\{1\},\0),\Pr(\{1\},\0),\Pr(\{1\},\0)) = h(\{1\},\{1\},\{1\})=I
\end{align*} where $J\in\I_C$ is arbitrary. The calculation shows that $I$ is a $\Psi_C$-answer set or, equivalently, a $\Psi_\Cb$-answer set and, clearly, it is the only one.

Now consider the slightly different program $C'=\{a\la; b\la a; c\la a; c\la b\}$. Then, $C'$ is isomorphically represented by the product $\Cb'=R\ltimes R\ltimes R\,[\I_{C'},\psi_{\Cb'}]$ given by
\begin{align*} 
\psi_{\Cb',1}(J) = \0 && \psi_{\Cb',2}(I_1,J) = \{1\}\sm I_1 && \psi_{\Cb',3}((I_1,I_2),J) = \{1\}\sm (I_1\cup I_2)
\end{align*} for all $I_1,I_2\in\Ir$ and $J\in\I_{C'}$. Let $h$ be defined as before. Iterating $\Psi_{\Cb'}$ bottom-up as above yields, for all $J\in\I_{C'}$:
\begin{align*} 
&h(\Psi_{\Cb'}((\0,\0,\0),J)) = h(\Pr(\0,\0),\Pr(\0,\{1\}),\Pr(\0,\{1\})) = h(\{1\},\0,\0) = \{a\} \\
&h(\Psi_{\Cb'}((\{1\},\0,\0),J)) = h(\Pr(\{1\},\0),\Pr(\0,\0),\Pr(\0,\0)) = h(\{1\},\{1\},\{1\}) = I \\
&h(\Psi_{\Cb'}((\{1\},\{1\},\{1\}),J)) = h(\Pr(\{1\},\0),\Pr(\{1\},\0),\Pr(\{1\},\0)) = h(\{1\},\{1\},\{1\}) = I
\end{align*} which shows that $I$ is also a $\Psi_{C'}$-answer set or, equivalently, a $\Psi_{\Cb'}$-answer set.
\end{example}

It is straightforward to generalize \yref{exa:C} to the general case.

\begin{theorem}\label{thm:positive} Every positive tight program $P$ can be isomorphically represented by a cascade program product of reset programs.
\end{theorem}

\section{Discussion and Conclusion}\label{sec:Discussion_and_Conclusion}

In this paper, we applied the Krohn-Rhodes Theory \cite{Krohn65}, presented here following \cite{Gecseg86}, to Answer Set Programming (ASP) \cite{Gelfond91}. Particularly, we defined a cascade product for ASP and, by relating programs and automata, showed that every program can be represented (or ``emulated'') by a product of very simple programs. We thus obtained nice theoretical results regarding the structure of ASP programs, which can be straightforwardly generalized to wider classes of nonmonotonic reasoning formalisms. More precisely, as our concepts and results hinge on the operator $\Pp$, they can be directly reformulated in the algebraic framework of Approximation Fixpoint Theory (AFT) \cite{Denecker00}, which captures, e.g., ordinary ASP, default and autoepistemic logic \cite{Denecker03}, and ASP with external sources \cite{Antic13}.

In a broader sense, this paper is a first step towards an algebraic theory of products and networks of nonmonotonic reasoning systems, including ASP and other formalisms. More precisely, we considered here only the very restricted (though powerful) kind of {\em cascade} product; it corresponds to the $\alpha_0$-product in \cite{Gecseg86}, and to the wreath product in finite semigroup theory \cite{Krohn65}. In the automata literature, however, many other important products have been studied (for an overview see \citeA{Domosi05}). We believe that recasting these kinds of products to ASP will lead to interesting results. Particularly, the notion of an asynchronous network \cite<cf.>[Chapter 7]{Domosi05} seems very appealing from an ASP point of view, as current modular ASP formalisms (e.g., \citeA{DaoTran09} cannot cope with asynchronous module structures according to our knowledge. Moreover, as different formalisms can be unified in the AFT-setting, heterogeneous networks in the vein of multi-context systems \cite<cf.>{Brewka11a} arise naturally. Finally, our concept of a product semantics {\em emerging} from the interaction of its simple factors (cf. \yref{sec:homomorphic}) seems interesting from a general AI perspective and we believe that it deserves a more intensive (and probably more intuitive) study in future work.

Although the Krohn-Rhodes decomposition theorem \cite{Krohn65} is now almost 50 years old, implementations and feasible applications of the Krohn-Rhodes Theory emerged only very recently \cite<cf.>{Egri-Nagy05}; reference 
our paper provides further evidence that it is a valuable tool for knowledge representation and reasoning in AI \cite<e.g.>{Egri-Nagy06}, and implementations in the ASP-setting remain as future work.

\bibliographystyle{theapa}
\bibliography{/Users/christianantic/Bibdesk/Bibliography,/Users/christianantic/Bibdesk/Preprints,/Users/christianantic/Bibdesk/Publications}
\end{document}